\documentclass[reqno, 11pt]{amsart}
\usepackage{fullpage}

\usepackage{enumerate}
\newcommand\bx{{\mathbf x}}

\newcommand\LL{{\mathbb L}}

\newcommand\N{{\mathbb N}}
\newcommand\R{{\mathbb R}}

\newcommand\Z{{\mathbb Z}}

\usepackage{amssymb, amsmath}

\newtheorem{theo}{Theorem}
\newtheorem{lemma}{Lemma}
\newtheorem{cor}{Corollary}
\newtheorem{assumption}{Assumption}

\title[]{SUPERDIFFUSIVITY OF ASYMMETRIC ENERGY MODEL
IN DIMENSION ONE AND TWO }

\author{C\'edric Bernardin}
\address{\!\!\!\!\!\!\!Universit\'e de Lyon, CNRS (UMPA)\newline
Ecole Normale Sup\'erieure de Lyon,\newline
46, all\'ee d'Italie,\newline
 69364 Lyon Cedex 07 - France.\newline
\rm {\texttt{Cedric.Bernardin@umpa.ens-lyon.fr}}\newline
\texttt{http://w3umpa.ens-lyon.fr/\~\;\!\!cbernard/}
}

\date{\today}

\thanks{\textsc{Acknowledgements. The author thanks F. Redig for valuable discussions and J. Fritz for useful comments on the existence problems of the dynamics. The author acknowledge the support of the French Ministry of Education through the ANR BLAN07-2184264 grant.} }

\keywords{Interacting particles system, super-diffusivity, KPZ class.}

\begin{document}


\maketitle

\begin{abstract}
We discuss an asymmetric energy model (AEM) introduced
by Giardina et al. in \cite{7}. This model is expected to belong to the KPZ
class. We obtain lower bounds for the diffusion coefficient. In particular,
the diffusion coefficient is diverging in dimension one and two as it is
expected in the KPZ picture.
\end{abstract}

\section{Introduction}
In their well-known paper (\cite{8}), M. Kardar, G. Parisi and Y.-C. Zhang
introduce a model for the evolution of the profile of a growing interface
\begin{equation}
\label{eq:1}
\partial_t h = \nabla^2 h +\cfrac{1}{2} (\nabla h)^2 +\eta (x,t)
\end{equation}

Here $h(x, t)\in \R$ is the height of the interface at location $x \in \R^d$ and at time
$t > 0$ and $\eta (x, t)$ is a space-time white noise. Starting from a flat state at
time $t = 0$, they are interested in the evolution of the fluctuations 
$$\xi(x, t) =< [h(x, t) -< h(x, t) >]^2 >^{1/2}$$
The intuitive picture is that the width grows with time as a power law up
to a saturation time that scales with the substrate size $L$ as $L^z$ where $z$ is
the scaling exponent. In other words, we expect that
$$\xi(L,t)= L^{2-z} f(t/L^z)$$
where the scaling function $f(x)$ saturates at large $x$ and $f(x) \sim x^{(2-z)/z}$
for $x\sim 0$ ([10]). By dynamic renormalization-group techniques, Kardar et
al. show that in dimension 1, the dynamic scaling exponent is $z = 3/2$
and $z = 2$ in dimension $d\geq 3$. The dimension $d = 2$ is the critical one
and their numerical studies indicate $1/z = 0.62 ± 0.04$. As noticed by
Kardar et al., equation (\ref{eq:1}) can be mapped to the Burgers equation
for a vorticity-free velocity field
\begin{equation}
\label{eq:3}
\partial_t v + v \cdot \nabla v = \nabla^{2} v -\nabla \eta (x,t)
\end{equation}
with $v = -\nabla h$ and $\eta (x, t)$ a space-time white noise.

Burgers equation is also closely connected to driven diffusive systems.
Consider a diffusive system under constant uniform driving force described
by a nonlinear Langevin equation. Then a quadratic order expansion of the
density gives the Burger's equation (\ref{eq:3}). In \cite{2}, H. van Beijern et al. investigate
the steady-state scattering function for driven diffusive systems with a
single conserved density. Mode-coupling arguments predict that in dimension
$d = 1$ (resp. $d = 2$), density fluctuations spread as $t^{2/3}$ (resp. $t (log t)^{2/3}$)
whereas they are of order $t^{1/2}$ (ordinary diffusion law) in dimension $d\geq 3$.
Guided by ideas of universality, we expect that a large class of microscopic
models whose evolutions are in a suitable coarse time and length scale an
approximation of the KPZ equation (\ref{eq:1}) (or the noisy Burgers equation) have a universal scaling exponent $z$ and a universal scaling limit.

Asymmetric Simple Exclusion Process (ASEP) is a natural discretization
(see \cite{16}) of the stochastic Burger's equation. The dynamics are given
by asymmetric random walks on $\Z^d$ with a drift in some direction such that jumps of
particles to occupied sites are forbidden (exclusion rule). It can be reinterpreted
as a growth model which is a natural discretization of the KPZ
equation. During the last decade, a lot of work has been accomplished to
test the validity of universality predictions for ASEP. In the one dimensional
case the value of the dynamical exponent $z = 3/2$ has been confirmed. Not only the exponent
but also the scaling function was obtained (\cite{6}). Moreover the limit is the same as the largest eigenvalue distribution in the random matrix theory (see \cite{4} and references therein for more informations). Nevertheless an important fact has to be mentioned: the methods used to obtain dynamical scaling exponent and limit distribution are very dependent of the
specific properties of ASEP. Indeed the main results are valid and proved
in the one dimensional case and for the Totally Asymmetric Simple Exclusion
Process (TASEP), which corresponds to nearest neighbors jumps in the
right direction. These results can not be carried for general ASEP. For example
M. Pr\"ahofer and H. Spohn (\cite{14}) compute the current fluctuations
for the TASEP but their proof does not work for other ASEP models. A more robust method has been introduced by  M. Bal\'azs and T. Sepp\"al\"ainen (\cite{1}) but is restricted to attractive systems. Even if they do not obtain the scaling limit function, they are able to establish the order of current for the nearest neighbors ASEP but not for general ASEP (obtained in \cite{16} by generalized duality techniques). Recently this method has been developed in the context of the Asymmetric Zero Range Process (AZRP) (see \cite{BK}).

To my knowledge, the only models belonging to the KPZ class and for
which a rigorous proof of the scaling order has been obtained are ASEP, AZRP,
Poly Nuclear Growth model (PNG) and related models (\cite{1}, \cite{BK}, \cite{5}, \cite{15}). In
dimension 2, the class is even more restrictive and the only rigorous result is
the scaling order obtained by H.T. Yau for general ASEP (\cite{19}). Hence, the
class of microscopic models for which one can rigorously prove they belong
to the KPZ class is very small and it is hence of extreme importance to have simple models for which one can rigorously prove they are in the KPZ class.

The aim of this paper is to study a \textit{non-attractive} model introduced by Giardina et al.
in \cite{7} and to show it presents anomalous behavior in low dimension as it is
expected in the KPZ picture.
In \cite{7}, Giardina et al. consider \textit{Symmetric Energy Model} (heat conduction
model in their terminology) and show the system has a dual process.

They also introduce an asymmetric generalization of the model that we call
the Asymmetric Energy Model (AEM). AEM should belong to the KPZ
universality class. It presents several analogies with ASEP but also differences
(\cite{7}). In this paper, we develop generalized duality properties for AEM
and obtain lower bounds for the bulk diffusion coefficient $D(t)$, i.e. the
variance of the two points correlation function (see (\ref{eq:12})). KPZ approach
predicts large time behavior of $D(t)$. In particular, $D(t)$ is expected to be divergent in dimension $1$ and $2$ and finite in dimension $d\geq 3$. The goal of this article is to obtain lower bounds for $D(t)$ consistent with this (theorem 1). The proof of this result is based on generalized
duality techniques introduced by C. Landim and H.T. Yau in the context
of ASEP. They have been developed in several directions but essentially for
lattice gas dynamics. Our main sources of inspiration are given by \cite{3} and \cite{12}. 

The paper is organized as follows. In section 2, we define AEM and introduce
the diffusion coefficient $D(t)$. Section 3 is devoted to the generalized
duality properties of the process. Section 4 contains the technical lemmas
necessary for the proof of theorem 1. The proofs of the main theorem are
given in section 5,6,7 for the $1$, $2$ and $d\geq 3$ case. The paper is ended by
remarks in section 8.

\section{The Asymmetric Energy Model (AEM)}

The system is composed of atoms indexed by $x \in \Z^d$. The canonical basis of ${\mathbb R}^d$ is denoted by $(e_1,\ldots,e_d)$. Each atom has a momentum $p_x \in \R$. Momenta are exchanged during the stochastic evolution in such a way that kinetic energy is conserved. The generator $L=S+A$ of AEM is defined by 
\begin{equation*}
(Sf)(p) =\sum_{i=1}^d \sum_{x \in \Z^d} (p_{x+e_i} \partial_{p_x} -p_x \partial_{p_{x+e_i}})^2
\end{equation*}
and
\begin{equation*}
(Af)(p)=\sum_{i=1}^d a_i \sum_{x \in \Z^d} p_x p_{x+e_i} (p_{x+e_i} \partial_{p_x} -p_x \partial_{p_{x+e_i}})
\end{equation*}
Here $p=(p_x)_{x \in \Z^d}$ is an element of the state space $\Omega=\R^{\Z^d}$ and $f$ is a smooth local function of $p$. Parameters $(a_i)_{i=1, \ldots,d}$ regulate the strength of the asymmetry in each direction. Let
\begin{equation*}
\mu_T (dp) = \otimes_{x \in \Z^d} (2\pi T)^{-1/2} \exp(-p_x^2 /2T) dp_x
\end{equation*}
be  the Gaussian product measure with temperature $T$. 
$\mu_T$ is an invariant probability measure for $L$. Moreover $S$ is symmetric and $A$ is antisymmetric in ${\mathbb L}^2 (\Omega, \mu_T)$. We fix now $T>0$ and denote $\mu_T$ by $<\cdot>$.
Energy of site $x$ is denoted by $E_x = p_x^2$. The formal total energy $\sum_{x \in \Z^d} E_x$ is a conserved quantity of the dynamics and one has
\begin{equation*}
L(E_x)= 2 \Delta(E_x) + 2\sum_{i=1}^d a_i \nabla_{e_i} (E_{x - e_i} E_x)
\end{equation*} 
where $\Delta$ is the usual $d$-dimensional discrete Laplacian:
\begin{equation*}
\Delta E_x = \sum_{i=1}^d \left\{ E_{x+e_i} +E_{x-e_i} -2E_x\right\}
\end{equation*}
and $\nabla_{e_i}$ is the discrete gradient in the direction $e_i$:
\begin{equation*}
\nabla_{e_i} E_x E_{x-e_i} = E_{x+e_i} E_x -E_x E_{x-e_i}
\end{equation*}
The microscopic instantaneous current in the direction $e_i$ is given by
\begin{equation*}
j_{x,x+e_i} (p) = 2 \nabla_{e_i} (E_{x}) +2 a_i E_{x}E_{x+e_i}
\end{equation*}
and one has the following microscopic continuity equation
\begin{equation*}
E_x (t) -E_x (0) = \sum_{i=1}^d \int_{0}^t \left(\nabla_{e_i} j_{x-e_i,x}\right)(p(s)) ds +M_x (t)
\end{equation*}
where $M_x (t)$ is a martingale.

We are interested in the energy-energy correlation function $S(x,t)$ defined by 
\begin{equation*}
S(x,t) = <E_x (t) E_0 (0)> - T^2
\end{equation*}

At equilibrium, the mean value of the current in the direction $e_i$ is $j_{i} (T)= \mu_T (j_{x,x+e_i})=2 a_i T^2$. By the conservation law (cf. \cite{18}, pp. 263-264), one has the two following informations for the average location and velocity of the structure function:
\begin{equation*}
\sum_x S(x,t)=2T^2
\end{equation*}
and
\begin{equation*}
\cfrac{1}{2T^2} \sum_x x S(x,t) =t v
\end{equation*}
where $v=\sum_{i=1}^d {j'}_i (T) e_i = 4T \sum_{i=1}^d a_i e_i$.

The third natural quantity to study is the bulk diffusion coefficient which is defined by
\begin{equation}
\label{eq:12}
D_{i,j} (t) =\cfrac{1}{4T^2 t} \left\{ \sum_{x \in \Z^d} x_i x_j S(x,t) -2T^2 (v_i t)(v_j t)\right\}
\end{equation}

Based on mode coupling theory (\cite{2}), it is expected that
\begin{equation*}
D(t) \sim 
\begin{cases}
t^{1/3}, \quad d=1\\
(\log t)^{2/3}, \quad d=2\\
1, \quad d \geq 3
\end{cases}
\end{equation*}
for large $t$.
Let $w_i (p)$ be the normalized current in the direction $e_i$
\begin{equation*}
w_i (p)= j_{0,e_i} (p) -j_i (T) - {j'}_i (T) (p_0^2 -T)
\end{equation*} 
or more explicitly
\begin{equation*}
w_i (p)= 2(1+a_i T) \nabla_{e_i} (E_0 ) +2a_i (E_0 -T) (E_{e_i} -T)
\end{equation*}
For local functions $f$ and $g$ in ${\mathbb L}^2 (\Omega)$, we define the semi-inner product
\begin{equation*}
\ll f, g \gg =\sum_{z \in \Z^d} (<\tau_z f g>-<f><g>)
\end{equation*}
Here $\tau_z$ is the usual shift on $\Omega$. The semi-norm corresponding to $\ll\cdot,\cdot\gg$ is denoted by $\| \cdot \|$. Note that (discrete) gradient terms $g=\tau_x h -h$ and constants vanish in this norm.

A formal integration by parts gives the following formula for the diffusion coefficient (see \cite{11})
\begin{equation}
\label{eq:diff}
D_{i,j}(t)= \cfrac{\delta_{i,j}}{2} + \cfrac{1}{4T^2} \ll {t^{-1/2} \int_{0}^t w_i (p(s))ds\, , \, t^{-1/2} \int_{0}^t w_j (p(s)) ds}\gg
\end{equation}

The Laplace transform of the diffusion coefficient is then given by
\begin{equation}
\label{eq:17}
\int_{0}^{\infty} e^{-\lambda t} t D_{i,j} (t) dt = \cfrac{1}{2\lambda^2} + \cfrac{1}{4 T^2 \lambda^2} \ll w_i, (\lambda -L)^{-1} w_j \gg
\end{equation}

We have to mention that  we don't have a rigorous proof of equality (\ref{eq:diff}). In fact even the existence of the equilibrium infinite volume dynamics is a non trivial problem (\cite{F1}). Several methods exist in the literature (\cite{F2}, \cite{OT} and references therein) but they are not directly applicable and we plan to extend these methods for AEM in a future work. In the rest of the paper, we assume the following: 

\begin{assumption}
The operator $L$ defined on the set of local integrable smooth functions of $\Omega$ is closable and its closure also denoted by $L$ is the generator of a strong Markov process. Moreover the set of local smooth integrable functions on $\Omega$ is a core for $L$. 
\end{assumption}

In the following RHS of (\ref{eq:diff}) will be used as \textit{definition} of $D(t)$. For ASEP, the validity of (\ref{eq:diff}) can be established by coupling techniques one can note translate in our context (see \cite{11} for more informations on this subject).
It follows that in terms of Laplace transform, behavior of $D(t)$ for large $t$ is, in a Tauberian sense, equivalent to behavior for small $\lambda$ of
\begin{equation*}
\ll w_i, (\lambda -L)^{-1} w_i \gg
\end{equation*}
For simplicity, we will restrict ourselves to the case $i=j$. We will prove the following theorem

\begin{theo}
\label{th:1}
There exists a constant $C>0$ such that
\begin{equation*}
\ll w_i, (\lambda -L)^{-1} w_i \gg \ge 
\begin{cases}
C\lambda^{-1/4}, \quad d=1\\
C |\log \lambda|^{1/2}, \quad d=2\\
C, \quad d\ge 3
\end{cases}
\end{equation*}
Moreover, if $d\ge 3$, we have also an upper bound
\begin{equation*}
\ll w_i , (\lambda -L)^{-1} w_i \gg \leq C^{-1}
\end{equation*}
\end{theo}

In a Tauberian sense, this theorem means that $D(t) \geq Ct^{1/4}$ for $d=1$, $D_{i,i} (t) \geq (\log t)^{1/2}$ for $d=2$, $C^{-1} \geq D_{i,i} (t) \geq C >0$ for $d\geq 3$.

Observe that the assumptions are not so relevant when we are only interested in lower and upper bounds. Without any assumption one can obtain similar bounds if we define the diffusion coefficient by a finite volume limit procedure. It means we can define $D(t)$ by the following limit (when it exists):
\begin{equation*}
D_{i,j} (t) =\lim_{N \to \infty} D_{i,j}^N (t), \quad D_{i,j} ^{(N)} (t)= \cfrac{1}{4T^2 t} \left\{ \sum_{x \in {\mathbb T}_N^d} x_i x_j S(x,t) -2T^2 (v_i t)(v_j t)\right\}
\end{equation*} 
where ${\mathbb T}_N^d$ is the $d$-dimensional discrete torus of length $N$ and the dynamics is now defined on ${\mathbb R}^{{\mathbb T}_N^d}$. Then one can prove the following lower bounds
\begin{equation*}
\lambda^{2} \int_{0}^{\infty} e^{-\lambda t} t D^{(N)}_{i,j} (t) dt \geq
\begin{cases}
C\lambda^{-1/4}, \quad d=1\\
C |\log \lambda|^{1/2}, \quad d=2\\
C, \quad d\ge 3
\end{cases}
\end{equation*}
where $C>0$ is independent of $\lambda$ and $N$. Moreover we have the corresponding upper bound if $d \geq 3$.   The advantage to deal directly with the infinite volume definition of the  diffusion coefficient is that it simplifies notations and avoid to work with discrete Fourier transform but directly with continuous Fourier transform.  

\section{Duality}
\label{sec:3}

For simplicity, we fix the temperature $T$ equal to one and we denote by $\mu$ or $<\cdot>$ the standard Gaussian product measure $\mu_1$. The Hilbert space associated to $\mu$ is denoted by ${\LL}^2 (\Omega)$ and the corresponding inner product is denoted by $<\cdot, \cdot>$. Let $h_n, n\ge 0$ be the sequence of Hilbert polynomials which is an orthogonal basis in ${\mathbb L}^2 ({\mathbb R}, (2\pi)^{-1/2} e^{-x^2 /2} dx)$:
\begin{equation*}
\int_{\mathbb R} h_n(x) h_m (x) \cfrac{e^{-x^2 /2}}{\sqrt{2\pi}}= \delta_{m,n} n!, \quad n,m \ge 0
\end{equation*}

The dual space of AEM is ${\mathbb N}^{{\mathbb Z}^d}$ and elements $\xi$ of ${\N}^{\Z^d}$ are seen as configurations of the "generalized dual process". If $\xi=(\xi_x\, ; \, x \in \Z^d )$ is a configuration of the dual space, we will say that $\xi$ is local if $\xi_x \neq 0$ only for a finite number of sites $x \in \Z^d$. In such a case the number of particles $|\xi|$ is defined by
\begin{equation*}
|\xi |=\sum_{x \in \Z^d} \xi_x
\end{equation*}
For any local configuration $\xi: \Z^d \to \N$, we define the multivariate Hermite polynomial function $H_{\xi}$ by
\begin{equation*}
H_{\xi} (p) =\prod_{x \in \Z^d} \cfrac{h_{\xi_x} (p_x)}{n(\xi_x)}
\end{equation*}
where $n:\N \to \N \backslash \{0\}$ is a suitable normalization function we will precise later. Remark that
\begin{equation*}
<H_\xi, H_\eta> = \cfrac{\xi !}{n^{2}(\xi)} \delta_{\xi,\eta}
\end{equation*} 
Here, the notations $\xi !$ and $n(\xi)$ are for $\prod_{x} (\xi_x !)$ and $\prod_x n(\xi_x)$.

Any local function $f \in \LL^2 (\Omega)$ can be decomposed uniquely as a finite linear combination of local functions
\begin{equation*}
f=\sum_\xi F(\xi) H_\xi = \sum_{n \ge 0} \sum_{|\xi|=n} F(\xi)H_\xi
\end{equation*}  
The coefficients of this linear combination are given by a real valued function $F$ defined on the dual space $\N^{\Z^d}$. Such a function is said to be of degree $n \ge 0$ if $F(\xi)=0$ as soon as $|\xi| \neq n$. For example, the normalized current $w_i (p)$ has the following decomposition
\begin{equation}
\label{eq:25}
w_i = 2 a_i \sum_\xi \delta_{\xi_i} (\xi) H_{\xi}\; + \; {\text{gradient term}}
\end{equation}  
where $\xi_i$ is the configuration with two particles on site $0$ and two particles on site $e_i$:
\begin{equation*}
(\xi_i )_x= 
\begin{cases}
2 \text{ if } x=0\\
2 \text{ if } x=e_i\\
0 \text{ otherwise}
\end{cases}
\end{equation*}
We know examine how acts the generator on the dual space. For this purpose we need to introduce notations. We introduce a cemetery  configuration $\circ$ which does not belong to the dual space. Any function $F : {\mathbb N}^{{\mathbb Z}^d} \to \mathbb R$ is extended to a function $F$ on ${\mathbb N}^{{\mathbb Z}^d} \cup \{\circ\}$ by $F(\circ)=0$. If $\xi$ is a configuration belonging to the dual space $\N^{\Z^d}$ and $x,y$ are two sites of $\Z^d$ then $\xi^{x,+2,y,-2}$ is the configuration obtained from $\xi$ by moving two particles from site $y$ to site $x$. If $\xi_y \le 2$ then $\xi^{x,+2,y,-2} = \circ$. The configuration $\xi^{x,+2}$ (resp. $\xi^{x,-2}$) is the configuration obtained from $\xi$ by adding (resp. by removing) two particles on site $x$. In the second case, if $\xi_x \le 2$ then $\xi^{x,-2} = \circ$. If $f \in \LL^2 (\Omega)$ is a smooth local function such that 
\begin{equation*}
f=\sum_{\xi \in {\mathbb N}^{{\mathbb Z }^d}} F(\xi) H_{\xi} = \sum_{n \ge 0} \sum_{|\xi|=n} F(\xi) H_\xi 
\end{equation*}
then 
\begin{equation*}
Lf= \sum_{i=1}^d \left\{ S_i f +a_i A_i f\right\}=\sum_{i=1}^d \left\{\sum_\xi ({\mathcal S}_i F)(\xi) H_\xi + a_i \sum_\xi ({\mathcal A}_i F)(\xi)H_\xi\right\}
\end{equation*}
We take
\begin{equation*}
n(k)=k!!=k(k-2)(k-4)\ldots, \quad n(0)=1
\end{equation*}
A long but elementary computation (see appendix) shows that
\begin{eqnarray*}
({\mathcal S}_i F)(\xi) &=& \sum_x (\xi_x +1) \xi_{x-e_i} \left[ F(\xi^{x,+2,x-e_i,-2})- F(\xi)\right]\\
&+& \sum_x (\xi_x +1) \xi_{x+e_i} \left[ F(\xi^{x,+2,x+e_i,-2})- F(\xi)\right]
\end{eqnarray*}
and
\begin{equation*}
{\mathcal A}_i = {\mathcal A}_i^0 + {\mathcal A}_i^+  +{\mathcal A}_i^-
\end{equation*}
with
\begin{equation*}
({\mathcal A}_i^0 F)(\xi) =\sum_x \left[ (\xi_{x+e_i} +1)\xi_x F(\xi^{x,-2,x+e_i, +2}) - (\xi_{x-e_i} -1) \xi_x F(\xi^{x,-2,x-e_i, +2})\right]
\end{equation*}
and
\begin{equation*}
({\mathcal A}_i^+ F)(\xi)= \sum_x \xi_x (\xi_{x+e_i} -\xi_{x-e_i}) F(\xi^{x,-2})
\end{equation*}
and
\begin{equation*}
({\mathcal A}_i^- F)(\xi) = \sum_x (\xi_x +1) (\xi_{x+e_i} -\xi_{x-e_i}) F (\xi^{x,+2})
\end{equation*}
Remark that ${\mathcal S}= \sum_{i=1}^d {\mathcal S}_i$ is the generator of a Markov process reversible       w.r.t. the measure
\begin{equation*}
m(\xi)= \cfrac{\xi !}{n^2 (\xi)}
\end{equation*}

This Markov process is not irreducible but if we restrict this process to the invariant subspace
\begin{equation*}
{\mathcal E}_{\varepsilon,n} = \left\{ \xi; \quad |\xi|=n,\; \xi_x =\varepsilon_x \text{ mod } 2\right\}, \quad \varepsilon \in \{0,1\}^{\Z^d}
\end{equation*}
then the restriction is irreducible. Let ${\mathcal H}_{\varepsilon, n}$ be the subspace of functions $F$ vanishing outside ${\mathcal E}_{\varepsilon, n}$.

Remark that the process corresponding to the generator $\mathcal S$ is the same as the dual process derived in \cite{7}. Nevertheless it is important to observe that our basis $\{H_\xi\}$ is different from the basis of \cite{7}. The key advantage of our choice is that the basis $\{H_\xi\}$ is orthogonal and the computations simplify considerably.

The operators ${\mathcal A}_i^0$, ${\mathcal A}_i^+$ and ${\mathcal A}_i^-$ are not Markov generators. The operator ${\mathcal A}_i^0$ conserves the degree and is antisymmetric in $\LL^2 (m)$. The operator ${\mathcal A}_i^+$ increases the degree by $2$ and ${\mathcal A}_i^-$ decreases the degree by $2$. In $\LL^2 (m)$, one has $({\mathcal A}_i^-)^*=-{\mathcal A}_i^+$.

We will use the following abusive but very convenient notations in the sequel. If $f,g \in {\LL}^2 (\Omega)$ are smooth local functions with coefficients in the basis $H_\xi$ given by local functions $F$ and $G$ then we have:
\begin{equation*}
<f,g>=<F,G>
\end{equation*}
with $<F,G>$ defined by
\begin{equation*}
<F,G>=\sum_{\xi \in \N^{\Z^d}} F(\xi) G(\xi) m(\xi)
\end{equation*}
We write similarly
\begin{equation*}
\ll f, g \gg = \ll F, G \gg
\end{equation*}
with 
\begin{equation}
\label{eq:40}
\ll F, G \gg = \sum_{x \in \Z^d} \sum_{\xi \in \N^{\Z^d} \backslash \{ 0 \} } F(\xi)G(\tau_x \xi) m(\xi)
\end{equation}
where $\tau_x \xi$ is the shifted configuration by $x$, meaning $(\tau_x \xi)_z=\xi_{x+z}$. Remark that in (\ref{eq:40}), the sum is carried over configurations $\xi$ with at least one particle. 

\section{Free particles Approximation}
\label{sec:4}

We introduce the generator $\Delta$ (discrete Laplacian) of independent random walks on ${\mathbb Z}^d$. It is given by:
\begin{equation*}
(\Delta G)(\eta) = \sum_{\substack{x,y \in \Z^d\\ |x-y|=1}} \eta_x \left[ G(\eta^{x,-1,y,+1}) -G(\eta)\right]
\end{equation*}
Here $\eta \in {\mathbb N}^{\Z^d}$ and $\eta^{x,-1,y,+1}$ is the configuration obtained from $\eta$ by moving a particle from site $x$ to site $y$. $G$ is a real valued local function defined on $\N^{\Z^d}$. A configuration $\eta$ can be seen as an element of $(\Z^d)/\Sigma_n$ where $\Sigma_n$ is the symmetric group of order $n$. The identification is given by the map
\begin{equation*}
(x_1, \ldots, x_n) \in (\Z^d)^n/\Sigma_n \to \eta
\end{equation*}
with 
\begin{equation*}
\eta (x)= \sum_{i=1}^n \delta_{x_i} (x)
\end{equation*}

A function $G$ in the domain of $\Delta$ is then identified with an element of ${\oplus}_{n=1}^{\infty} {\LL^2} ((\Z^d)^n/\Sigma_n)$ (with respect to the counting measure).
We denote the standard inner product on ${\oplus}_{n=1}^{\infty} {\LL^2} ((\Z^d)^n/\Sigma_n)$ by $<\cdot, \cdot>_0$:
\begin{equation}
\label{eq:48}
\begin{array}{lcl}
\forall F,G \in {{\mathbb L}^2} ((\Z^d)^n/\Sigma_n),  \quad <F,G>_0 &=& \sum_{\bx =(x_1,\ldots,x_n) \in (\Z^d)^n/\Sigma_n} F(\bx)G(\bx)\\
&=& \cfrac{1}{n!} \sum_{\bx= (x_1,\ldots,x_n) \in (\Z^d)^n} F(\bx)G(\bx)
\end{array}
\end{equation}

We define also the scalar product with translations associated to $<\cdot,\cdot>_0$
\begin{equation*}
\ll F, G \gg_0 = \sum_{x \in \Z^d} \left( <F , \tau_x G>_0 -<F>_0<G>_0\right)
\end{equation*}

We recall that $\Delta$ is a positive self-adjoint operator w.r.t. $<\cdot,\cdot>_0$.

\begin{lemma}
\label{lem:1}
Let $F(\xi)$ be a local function belonging to ${\mathcal H}_{0,2n}$ and define $G(\eta)=F(2\eta)$ then
\begin{equation*}
C^{-1} 2^{-2n} <-\Delta G, G>_0 \le <-{\mathcal S}F, F > \le C n <-\Delta G, G>_0
\end{equation*}
where $C$ is a positive constant independent of $n$ and $F$.
\end{lemma}

\begin{proof}
One has
\begin{equation*}
<-{\mathcal S} F, F> = \sum_{\xi \in {\mathcal E}_{0,n}} \sum_{|x-y|=1} \xi_x (1+\xi_y)\left[ F(\xi^{x,-2,y,+2}) -F(\xi) \right]^2 m(\xi)
\end{equation*}

For any integer $k$, we have

\begin{equation*}
m(2k)= \cfrac{(2k)!}{(2k!!)^2}=\cfrac{\prod_{j=1}^k (2j -1)}{\prod_{j=1}^k 2j}=\prod_{j=1}^k \left( 1-\cfrac{1}{2j}\right)
\end{equation*}
and
\begin{equation*}
m(2k+1)= \cfrac{(2k+1)!}{((2k+1)!!)^2}=\cfrac{\prod_{j=1}^k (2j)}{\prod_{j=1}^k (2j+1)}=\prod_{j=1}^k \left( 1-\cfrac{1}{2j+1}\right)
\end{equation*}
so that $2^{-k} \le m(2k) \le 1$ and $(2/3)^k \le m(2k+1) \le 1$. It follows that
\begin{equation*}
2^{-2n} \le \inf_{\xi \in {\mathcal E}_{0,2n}} m(\xi) \le \sup_{\xi \in {\mathcal E}_{0,2n}} m(\xi) \le 1
\end{equation*}
We use now the fact that $1 \le (1+\xi_y) \le (1+2n)$ and we conclude.
\end{proof}

It follows that the same lemma is also true for the inner product with translations since
\begin{equation}
\label{eq:trans}
\ll F,G \gg_{(0)} =\lim_{k \to \infty} \cfrac{1}{(2k+1)^d} \left< \left(\sum_{|x| \leq k} \tau_x F\right), \left(\sum_{|x| \leq k} \tau_x G \right) \right>_{(0)}
\end{equation}

\begin{lemma}
\label{lem:2}
Let $F(\xi)$ a local function belonging to ${\mathcal H}_{0,2n}$ and define $G(\eta)=F(2\eta)$. There exists a constant $C$ independent of $n$ and $F$ such that
\begin{equation*}
C^{-1} 2^{-2n} \ll -\Delta G, G \gg_0 \le \ll -{\mathcal S}F, F \gg \le C n \ll -\Delta G, G\gg_0
\end{equation*}
\end{lemma}

The $H_{-1,\lambda}$ norms are defined by
\begin{equation*}
\| F \|_{-1,\lambda}^2 = \ll F, (\lambda -{\mathcal S} )^{-1} F \gg = \sup_{G} \left\{ 2 \ll F, G \gg - \ll G, (\lambda -{\mathcal S}) G \gg\right\}
\end{equation*}
and
\begin{equation*}
\|F\|_{-1,\lambda,0}^2 = \ll F, (\lambda -\Delta)^{-1} F \gg_0 = \sup_{G} \left\{ 2 \ll F,G \gg_0 - \ll G, (\lambda -\Delta)G \gg_0 \right\}
\end{equation*}
In these formulas, function $F$ is a local function from $\N^{\Z^d}$ into $\R$ and the supremum is carried over local functions.

It follows easily from lemma $2$ and equation (\ref{eq:48}) that there exists a positive constant $C(n)$ such that for every local square integrable function $F$ belonging to ${\mathcal H}_{0,2n}$
\begin{equation*}
\cfrac{1}{C(n)} \|mF\|_{-1,\lambda,0}^2 \le \|F \|_{-1,\lambda}^2 \le C(n) \|mF\|_{-1, \lambda,0}^2
\end{equation*}

In this inequality, the function $mF$ is defined by $(mF)(\xi)= m(\xi) F(\xi)$.

Every function belonging to to ${\mathcal H}_{0,2n}$ can be seen as a symmetric function ${\tilde F}$ from $(\Z^d)^n$ into $\R$:
\begin{equation*}
{\tilde F} (x_1,\ldots,x_n)= F(2\delta_{x_1} + \ldots +2\delta_{x_n})
\end{equation*}
where $\delta_x$ denotes the configuration with only one particle on site $x$. Sum of two configurations and multiplication by an integer are defined in the standard way. In the sequel, we will identify $F$ with ${\tilde F}$.

To obtain a lower bound, we use the following variational formula for the Laplace transform (see \cite{3}):
\begin{equation}
\label{eq:55}
\begin{array}{l}
\ll w_i, (\lambda -L)^{-1} w_i \gg\\
= \sup_f \left\{ 2 \ll f,w_i \gg - \ll f, (\lambda -S)f \gg - \ll Af, (\lambda -S)^{-1} Af \gg \right\}\\
=\sup_F \left\{ 2\ll F, \delta_{\xi_i}\gg -\| F \|_{1,\lambda}^2 -\| {\mathcal A}F \|_{-1,\lambda}^2 \right\}
\end{array}
\end{equation}
and we restrict the supremum over functions $f=\sum_\xi F(\xi) H_\xi$ such that $F$ belongs to ${\mathcal H}_{0,4}$. 

\begin{lemma}
Let $d\geq 1$ and $F,G : (\Z^d)^n \to \R$ be symmetric local functions. There exists a constant $C:=C(d,n)$ such that 
\begin{equation*}
\left| <G,{\mathcal A}^0 F >\right| \le C \left<G, (\lambda - {\mathcal S} G\right>^{1/2} \left< F, (\lambda -{\mathcal S})F \right>^{1/2}
\end{equation*}
The same inequality is valid for the inner product $\ll \cdot,\cdot \gg$:
\begin{equation*}
\left| \ll G,{\mathcal A}^0 F \gg \right| \le C \ll G, (\lambda - {\mathcal S}) G \gg^{1/2} \ll F, (\lambda -{\mathcal S})F \gg^{1/2}
\end{equation*}
\end{lemma}

\begin{proof}
By definition of $m$, for any nonnegative integers $k,\ell$, we have
\begin{equation}
\label{eq:58}
m(k+2) (k+2) m(\ell -2) (\ell -2) =(k+1) m(k) \ell m(\ell)
\end{equation}
and by simple changes of variables, one has:
\begin{equation*}
\begin{array}{lcl}
\left< {\mathcal A}_i^0 F,G\right>&=& \cfrac{1}{2} \sum_{\xi,x} \xi_x (\xi_{x+e_i} +1) \left[ F(\xi^{x-2,x+e_i, +2}) -F(\xi)\right]\left[ G(\xi)- G(\xi^{x,-2,x+e_i,+2})\right]m(\xi)\\
&-&\cfrac{1}{2} \sum_{\xi,x} \xi_x (\xi_{x-e_i} +1) \left[ F(\xi^{x-2,x-e_i, +2}) -F(\xi)\right]\left[ G(\xi)- G(\xi^{x,-2,x-e_i,+2})\right]m(\xi)
\end{array}
\end{equation*}
The estimate is then a simple consequence of Schwarz's inequality. The second part of the lemma follows from the definition of $\ll \cdot, \cdot \gg$.
\end{proof}

\begin{cor}
\label{cor:1}
Let $d \geq 1$ and $F:(\Z^d)^n \to \R$ be a symmetric local function. There exists a constant $C=C(d,n)>0$ such that
\begin{equation*}
\| {\mathcal A}^0F\|_{-1,\lambda}^2 \leq C \| F\|_{1,\lambda}^2
\end{equation*} 
\end{cor}
Recall we want to obtain a lower bound for the formula (\ref{eq:55}). Corrolary \ref{cor:1} implies we can forget the term $\| {\mathcal A}^0 F\|_{-1,\lambda}^2 $ in the variational formula (\ref{eq:55}).

\section{The one dimensional case}
\label{sec:5}

We recal that every function $F$ belonging to ${\mathcal H}_{0,2n}$ is identified with a symmetric function ${\tilde F}$ from $\Z^n$ into $\R$:
\begin{equation*}
{\tilde F}(x_1, \ldots, x_n) =F(2 \delta_{x_1} + \ldots+ 2\delta_{x_n})
\end{equation*}

In the sequel, we forget the $\tilde {\phantom F}$ and consider a function $F \in {\mathcal H}_{0,n}$ as a symmetric function on ${\Z}^n$.

To obtain a lower bound, we restrict the supremum over functions $f=\sum_\xi F(\xi) H_\xi$ such that $F$ belongs to ${\mathcal H}_{0,4}$. We identify hence $F$ with a symmetric function on $\Z^2$.

By corollary \ref{cor:1}, $\|{\mathcal A}^0 F\|_{-1,\lambda}^2$ is of no importance and can be forgotten. 

Since $F \in {\mathcal H}_{0,4}$, we have that ${\mathcal A}^+ F$ belongs to ${\mathcal H}_{0,6}$. ${\mathcal A}^+ F$ is identified with a symmetric function on $\Z^3$.
 
 For all symmetric function $F(x,y)$ (resp. $G(x,y,z)$) from $\Z^2$ (resp. $\Z^3$) to $\Z^2$ (resp. $\Z^3$), we define
 \begin{equation*}
 F^* (\alpha)= \sum_{z\in \Z} F(z,z+\alpha)
 \end{equation*}
resp.
\begin{equation*}
G^* (\alpha, \beta)= \sum_{z \in Z} G(\alpha +z, \beta+z,z)
\end{equation*}
One can check that
\begin{equation*}
\ll F,F \gg_0 =\sum_{\alpha \in \Z}( F^* (\alpha))^2, \quad \ll G,G \gg_0 =\sum_{\alpha, \beta \in \Z} (G^* (\alpha,\beta))^2 
\end{equation*}
and 
\begin{equation}
\label{eq:67}
\ll F, (-\Delta F) \gg_0 =\sum_{u \in \Z} \sum_e \left[ F^* (u+e) -F^* (u)\right]^2, \quad \ll G, (-\Delta G) \gg_0 =\sum_{u \in \Z^2} \sum_e \left[ G^* (u+e) -G^* (u)\right]^2
\end{equation}
where the sum over $e$ is carried over all $e \in \Z^2$ sucht that $e=\pm(0,1),\pm(1,0), \pm (1,1)$.

The advantage of this notation is to reduce the degree of functions. We take now a function $F \in {\mathcal H}_{0,4}$ (i.e. a symmetric function on $\Z^2$) and we compute $G^* = (m {\mathcal A}^+ F)^*$. One has

\begin{equation*}
\begin{array}{l}
(m{\mathcal A}^+ F) (x,x,x+1) = \cfrac{3}{2} (F(x,x+1)-F(x,x))\\
(m {\mathcal A}^+ F) (x,x,x-1)= \cfrac{3}{2} (F(x,x) - F(x-1,x))\\
(m {\mathcal A}^+ F)(x-1,x,x+1) = \cfrac{1}{2} (F(x,x+1)- F(x-1,x))\\
(m {\mathcal A}^+ F) (x,x+1,y) = \cfrac{1}{2} (F(x+1,y)- F(x,y)), \quad y\neq x-1,x,x+1,x+2\\
(m{\mathcal A}^+ F)(x,y,z)= 0 \quad \text{otherwise}
\end{array}
\end{equation*}

A long but elementary computation shows that:
\begin{equation}
\label{eq:G}
\begin{array}{lcl}
G^{*} (x,y) &=& \cfrac{1}{4} \left[\delta_1 (y-x) - \delta_1 (x-y) \right](F^* (y)-F^*(x))\\
&-&\cfrac{1}{4} \delta_1 (|x|) \delta_1 (y-x) (F^* (x+1)- F^* (x))\\
&-&\cfrac{1}{4} \delta_1 (|y|) \delta_1 (x-y) (F^* (y+1)- F^* (y))\\
&+&\cfrac{3}{2} ( \delta_1 (x)\delta_1 (y) -\delta_{-1}(x) \delta_{-1}(y))(F^* (0) -F^* (1))\\
&+&\cfrac{1}{4} \delta_{-1} (x) \delta_{\ge 2} (y) (F^* (y-1) -F^* (y))\\
&+&\cfrac{1}{4} \delta_{-1} (y) \delta_{ \ge 2} (x) (F^* (x-1) -F^* (x))\\
&+& \cfrac{1}{4} \delta_{1} (x) \delta_{\ge 3} (y) (F^* (y)-F^* (y+1))\\
&+&\cfrac{1}{4} \delta_1 (y) \delta_{\ge 3} (x) (F^{*}(x) -F^* (x+1))
\end{array}
\end{equation}

We choose now the following test function $F$:
\begin{equation}
\label{eq:70}
F^* (x)= \lambda^{-1/4} e^{-\lambda^{3/4} |x|}
\end{equation}
and evaluate for this test function the terms appearing in the variational formula (\ref{eq:55}).  In the sequel, $C$ denotes a positive constant independent of $\lambda$ which can change from line to line. 

By using the decomposition (\ref{eq:25}) of $w_i$, we have that for small $\lambda >0$:
\begin{equation*}
\ll \delta_{\xi_1}, F \gg =2a_1 \sum_{x} F(x,x+1)=2a_1 F^* (1) \sim \lambda^{-1/4}
\end{equation*}
Moreover, since $F$ belongs to ${\mathcal H}_{0,4}$, there exists a positive constant $C$ such that
\begin{equation*}
C^{-1} <F,F>_0 \le <F,F> \le C<F,F>_0 
\end{equation*}
and by (\ref{eq:trans}), the same is true for the inner-products $\ll \cdot, \cdot \gg_0$ / $\ll \cdot, \cdot \gg$. The norm of $F$ with respect to $\ll \cdot, \cdot \gg_0$ is easy to evaluate and is of order $\lambda^{-1/4}$. 
The third term to estimate is $\ll F, -{\mathcal S} F \gg$. By lemma \ref{lem:2}, this term is of the same order as $\ll F, -{\mathcal S} F \gg_0$. Thanks to formula (\ref{eq:67}) and explicit form of $F^*$ , we obtain that $\ll F, (- {\mathcal S}) F \gg$ is of order $\lambda^{1/4}$. Hence we proved
\begin{equation}
\label{eq:71}
2 \ll \delta_{\xi_1}, F \gg \sim \lambda^{-1/4}, \quad \lambda \ll F, F \gg \sim \lambda^{-1/4}, \quad \ll F, -{\mathcal S} F \gg \sim \lambda^{1/4}
\end{equation}

 We now evaluate the last term
 \begin{equation*}
 \ll [m({\mathcal A}^+ F)], (\lambda -\Delta)^{-1} [m({\mathcal A}^+ F)] \gg_0
 \end{equation*}
 
 \begin{lemma}
 \label{lem:4}
 Let $F^* (x)= \lambda^{-1/4} e^{-\lambda^{3/4} |x|}$. We have
 \begin{equation}
 \ll [m({\mathcal A}^+ F)], (\lambda -\Delta)^{-1} [m({\mathcal A}^+ F)] \gg_0 \sim \lambda^{-1/4}
 \end{equation}
 \end{lemma}
 
 \begin{proof}
 For any function $f:\Z^n \to \R$, we introduce the Fourier transform ${\hat f}$ of $f$ defined by 
 \begin{equation*}
 {\hat f}(\xi) =\sum_{x \in \Z^n} f(x) e^{2i\pi x \cdot \xi}, \quad \xi \in [0,1]^n
 \end{equation*}
 By (\ref{eq:67}), we have
 \begin{equation*}
 \ll [m({\mathcal A}^+ F)], (\lambda -\Delta)^{-1} [m({\mathcal A}^+ F)] \gg_0 = \int_{[0,1]^2} ds dt \cfrac{|{\hat G^*} (s,t)|^2}{\lambda + 4 \theta (s) + 4\theta (t) + 4 \theta (s+t)}
 \end{equation*}
 where $\theta (u)= \sin^2 (\pi u)$. We have
 \begin{equation*}
 \ll [m({\mathcal A}^+ F)], (\lambda -\Delta)^{-1} [m({\mathcal A}^+ F)] \gg_0 \leq \int_{[0,1]^2} ds dt \cfrac{|{\hat {G^*}} (s,t)|^2}{\lambda + 4 \theta (s) + 4\theta (t)}= \left< G^*, (\lambda -\Delta)^{-1} G^* \right>_0
 \end{equation*}
 We can express Fourier transform of $G^*$ in terms of ${\hat F}^*$. We write $G^* = G_1^* +G_2^*+\ldots+G_8^*$ with $G_1^*, \ldots, G_8^*$ the eight terms appearing in (\ref{eq:G}). We claim that
 \begin{equation*}
 \begin{cases}
 \ll G_j, (\lambda -\Delta)^{-1} G_j \gg_0 \sim \lambda^{-1/4}, \; j=1,5,6,7,8,\\
 \ll G_j, (\lambda-\Delta)^{-1} G_j \gg_0 =O(\lambda \log \lambda), \; j=2,3,4
 \end{cases}
 \end{equation*}

 We begin by the proof of the first claim. We have
 \begin{equation*}
 G_1^* (x,y)= \cfrac{1}{4} \left[ \delta_1 (y-x) - \delta_1 (x-y)\right] (F^* (y) -F^* (x))
 \end{equation*} 
and 
 \begin{equation*}
 {\hat G}_1^* (s,t) = -\cfrac{i}{2} \left\{ \sin (2\pi s) +\sin (2\pi t) \right\} {\hat F}^* (s+t)
 \end{equation*}
 The Fourier transform of ${\hat F}^*$ is easy to compute and is given by
 \begin{equation*}
 {\hat F}^* (s) = \cfrac{ \lambda^{-1/4} (1-2e^{-2\lambda^{3/4}})}  {(1-e^{-\lambda^{3/4}})^2 + 4e^{-\lambda^{3/4}} \sin^{2} (\pi s)}
 \end{equation*}
 and we have
 \begin{equation*}
 | {\hat F}^* (s) | \leq \cfrac{C {\sqrt \lambda}}{\lambda^{3/2} +\sin^2 (\pi s)}
 \end{equation*}
 as soon as $\lambda$ is sufficiently small. Moreover, a simple computation shows that
 \begin{equation}
 \label{eq:82}
 \begin{array}{lcl}
 \int_{[0,1]^2} ds dt \cfrac{|{\hat G}_1^* (s,t)|^2}{\lambda +4\theta (s) +4 \theta(t)}&=& \int_{[0,1]^2} ds dt \cfrac{(\sin (2\pi s) +\sin (2\pi t))^2 |{\hat F}^* (s+t)|^2}{\lambda + 4\theta (s) +4\theta (t)}\\
 &=& \int_0^1 ds |{\hat F}^* (s)|^2 \sin^2 (\pi s) U_\lambda (s) 
 \end{array}
 \end{equation}
  where
 \begin{equation*}
U_{\lambda} (s) =\int_0^1 \cfrac{\cos^2 (\pi u)}{\lambda + 2 \sin^2 \left( \cfrac{\pi (u+s)}{2}\right) +2 \sin^2 \left(\cfrac{\pi (u-s)}{2} \right)}du
 \end{equation*}
 We are interested in the behavior of (\ref{eq:82}) as $\lambda$ goes to zero. The critical points of the integrand are $0$ and $1$ and by symmetry arguments, we can restric ourselves to $0$. By equation (4.12) of \cite{3}, one has
 \begin{equation*}
 U_{\lambda} (s) \leq \cfrac{C}{\sqrt{\lambda +s^2}}
 \end{equation*} 
 
 It follows that
 \begin{equation*}
 4\int_0^{1/2} ds |{\hat F}^* (s)|^2 \sin^2 (\pi s) U_\lambda (s) \leq C\lambda  \int_0^{1/2} ds \cfrac{\sin^2 (\pi s)}{(\lambda^{3/2} +\sin^2 (\pi s))^2 \sqrt{\lambda +s^2}}
 \end{equation*}
 Standard analysis shows that this last term is of order $\lambda^{-1/4}$ which proves the first claim. 
 
 Let us examine the term $G_8$ (terms $G_5, G_6,G_7$ are evaluated in the same way). We have
 \begin{equation*}
 {\hat G}_8^* (x,y)= \cfrac{\lambda^{-1/4}(1-e^{-\lambda^{3/4}})}{4} \delta_1 (y) \delta_{x\geq 3} e^{-\lambda^{3/4} x}
 \end{equation*}
By a direct computation, we obtain
\begin{equation*} 
 |{\hat {G_8^*}} (s,t)|^2 \leq C \cfrac{\lambda}{\lambda^{3/2} + \sin^2 (\pi s)}
 \end{equation*}
 
 It follows that
 \begin{equation*}
 \begin{array}{lcl}
 \ll G_8 , (\lambda -\Delta)^{-1} G_8 \gg_0 &=& \int_{[0,1]^2} ds dt \cfrac{|{\hat G}_8^* (s,t)|^2}{\lambda +4 \theta(s) + 4\theta (t) +\theta (s+t)}\\
 &\leq& C \lambda \int_0^1 ds \cfrac{1}{\lambda^{3/2} +\sin^2 (\pi s)} \left(\int_0^1 \cfrac{dt}{\lambda +\theta(s) +\theta (t)}\right)\\
 &\leq& C \lambda \int_{0}^{1} \cfrac{ds}{(\lambda^{3/2} +s^2)(\sqrt{\lambda +s^2})}\\
 &=& O(\lambda^{-1/4}) 
 \end{array}
 \end{equation*}
because we have
\begin{equation*}
\int_{0}^1 \cfrac{dt}{\lambda +\theta(s) +\theta (t)} \leq \cfrac{C}{\sqrt{\lambda +\theta(s)}}
\end{equation*}
 
For the terms $G_2, G_3, G_4$, we use the fact that $|{\hat{G_j^*}} (s,t)| \leq C \lambda^{1/2}$ if $j=2,3,4$ and we observe that
\begin{equation*}
\int_{[0,1]^2} \cfrac{ds dt}{\lambda +\theta (s) +\theta(t)} = O (\log \lambda)
\end{equation*}
 
 The lemma is proved.
 \end{proof}
 
 We can now conclude the proof of theorem \ref{th:1} for the one dimensional case. We use the variational formula (\ref{eq:55}) and the test function
 \begin{equation*}
 f(p)=a\sum_{x,y} F(x,y) H_{2\delta_x +2\delta_y} (p)
 \end{equation*} 
Here $a$ is a positive constant we will fix later and the function $F$ is defined in (\ref{eq:70}). By (\ref{eq:71}), lemma \ref{lem:4} and equation (\ref{eq:55}), we obtain
 \begin{equation*}
 \ll w_1, (\lambda -L)^{-1} w_1 \gg \ge C_1 a \lambda^{-1/4} -C_2 a^2 \lambda^{1/4} -C_3 a^{2}\lambda^{-1/4}
 \end{equation*}
 where $C_1, C_2, C_3$ are positive constants independent of $a$ and $\lambda$. If $a$ is chosen sufficiently small, the lower bound is of order $\lambda^{-1/4}$ and theorem \ref{th:1} is proved.
 
\section{The $2$-dimensional case}
\label{sec:6}

For simplicity, we assume $a_2=0$ and $a_1=1$.

\begin{lemma}
\label{lem:5}
Fix $R>0$ and assume $d \ge 2$. There exists a positive constant $C=C(d)$ independent of $\lambda$ and $n$ such that for any local symmetric function $F: (\Z^d)^n \to \R$
\begin{equation*}
\sum_{i,j,k,l=1}^n \left< {\bf 1}_{|\bx_i - \bx_j |+ | \bx_k -\bx_l | \le R }F^2\right>_0 \le Cn^3 \left< F, (\lambda -\Delta) F \right>_0
\end{equation*}
\end{lemma}
 
 \begin{proof}
 This lemma is proved in lemma 4.2 of \cite{12}.
 \end{proof}
 
 \begin{lemma}
 \label{lem:6}
 Let $d \ge 2$ and $n \ge 1$. There exists a constant $C(n)$ such that for any local symmetric functions $F: (\Z^d)^{n \mp 1} \to \R$ and $G: (\Z^d)^n \to \R$
 \begin{equation*}
 \left| \left< G, {\mathcal A}^{\pm} F\right> - \left< G, {\mathcal A}^{\pm} F\right>_0\right| \le C(n) \left<F, (\lambda -\Delta)F \right>_0^{1/2} \left<G, (\lambda -\Delta)G \right>_0^{1/2}
 \end{equation*}
 and it follows that
 \begin{equation*}
 \left| \ll G, {\mathcal A}^{\pm} F\gg - \ll G, {\mathcal A}^{\pm} F\gg_0\right| \le C(n) \ll F, (\lambda -\Delta)F \gg^{1/2} \ll_0 G, (\lambda -\Delta)G \gg_0^{1/2}
 \end{equation*}
 \end{lemma}
 
 \begin{proof}
 One has
 \begin{equation*}
 ({\mathcal A}_i^+ F) (\bx) = \sum_{k,\ell =1}^n {\bf 1}_{\{ \bx_k =\bx_\ell +e_i \}} \left( F (\bx^k) -F (\bx^\ell)\right)
 \end{equation*}
 with $\bx^k$ the vector $\bx$ where the coordinate $\bx^k$ has been removed. Observe it is the operator defined in formula (3.4) of \cite{12}. We have
 \begin{equation}
 \label{eq:94}
 \left< G, {\mathcal A}^{+}_i F\right> - \left< G, {\mathcal A}_i^{+} F\right>_0
 = \sum_{ k, \ell =1}^n {\bf 1}_{\{ \bx_k =\bx_\ell + e_i\}} G(\bx) (m(\bx) -1) (F(\bx^k) - F(\bx^\ell))
 \end{equation} 
 If $\bx \in (\Z^d)^n$ is such that $\bx_i \neq \bx_j$ for every $i \neq j$ then $m (\bx) -1 =0$. Hence we can introduce the sum of the following indicator function
 \begin{equation*}
 \sum_{i \neq j} {\bf 1}_{\bx_i = \bx_j}
 \end{equation*}
 in the sum (\ref{eq:94}). We recall that $| m(\bx)-1| \le C(n)$ for a constant $C(n)$. It follows that
 \begin{eqnarray*}
 \left|<G, {\mathcal A}^+_i F> - <G, {\mathcal A}_i^+ F>_0 \right|\\
 \le C(n) \sum_{i,j,k,\ell} \sum_\bx \left| G(\bx) {\bf 1}_{|\bx_k - \bx_\ell | \le 1}{\bf 1}_{| \bx_i - \bx_j | \le 1} \right| |F(\bx^k) - F(\bx^\ell) |
 \end{eqnarray*}
 One concludes by Schwarz inequality.
 \end{proof}
 
 \begin{cor}
 \label{cor:2}
 Let $d\ge 2$ and $F: (\Z^d)^n \to \R$. There exists a constant $C(n)$ such that
 \begin{equation*}
 \| {\mathcal A}^{\pm} F\|_{-1,\lambda}^2 \le C(n) \left( \| {\mathcal A}^{\pm} F \|_{-1,\lambda, 0}^2 + \| F \|_{1,\lambda,0}^2 \right)
 \end{equation*}
 \end{cor}
 
\begin{proof}
We have only to prove the lemma with the inner product $\ll \cdot, \cdot \gg$ replaced with the standard inner product without translations $<\cdot,\cdot>$ thanks to (\ref{eq:trans}). Recall that
\begin{eqnarray*}
\left<  {\mathcal A}^{+} F, (\lambda -{\mathcal S})^{-1} {\mathcal A}^+ F\right>\\
= \sup_{G} \left\{ 2<\mathcal A^+ F, G> -\lambda <G,G>- <G, (- {\mathcal S} )G> \right\}
\end{eqnarray*}
Then we use lemma \ref{lem:6}, lemma \ref{lem:1} and equation (\ref{eq:48}) to conclude.
\end{proof} 
 
 We can now complete the proof. Recall that we want obtain a lower bound for the RHS of (\ref{eq:55}) and remark  that in the dual basis, one has
 \begin{equation*}
 \ll w_1, f \gg = 2a_1 \sum_{z \in \Z^d} F(2\delta_z + 2\delta_{z+e_1})
 \end{equation*}
 where $F$ is the function corresponding to $f$ in the dual basis. we restrict the supremum over functions belonging to ${\mathcal H}_{0,4}$. Hence, by corrolary {\ref{cor:1}}, corrolary {\ref{cor:2}} and lemma \ref{lem:2}, one has:
 \begin{eqnarray*}
 \ll w_1, (\lambda -L)^{-1} w_1 \gg \\
 \ge \sup_{F \in {\mathcal H}_{0,4}} \left\{ 4a_1 \sum_z F(2z, 2z + 2e_1) -C \| F\|_{1,\lambda,0}^2 - C \| {\mathcal A}^+ F\|_{-1, \lambda,0}^2\right\}
 \end{eqnarray*}
 The maximizer for this new variational problem can be computed using Fourier transform (cf. lemma 3.3 of \cite{12}). It turns out that
 \begin{equation*}
 \ll w_1, (\lambda -L)^{-1} w_1 \gg \ge 
 \begin{cases}
 C|\log \lambda|^{1/2} \quad d=2\\
 C \quad d \ge 3
 \end{cases}
 \end{equation*}
 
\section{The $d$-dimensional case for $d \geq 3$}
\label{sec:7}
In this section, we show that the diffusion coefficient is finite and strictly positive if the dimension $d$ is greater than $3$. For the asymmetric simple exclusion process, this has been proved in \cite{13}. The lower bound follows from the previous section. 
The upper bound is obtained by ignoring the asymmetric part of the generator:
\begin{eqnarray}
\label{eq:102}
\ll w_1 , (\lambda - L)^{-1} w_1 \gg\\
=\sup_{f} \left\{2 \ll f, w_1 \gg - \ll f, (\lambda -S)f \gg - \ll Af, (\lambda -S)^{-1} Af \gg \right\} \nonumber \\
\le \sup_f  \left\{2 \ll f, w_1 \gg - \ll f, (\lambda -S)f \gg \right\} \nonumber
\end{eqnarray}
We write this last variational formula in the dual basis $H_\xi$ and we recall that
\begin{equation*}
\ll w_1, f \gg = 2a_1 \sum_{z \in \Z^d} F(2\delta_z + 2\delta_{z+e_1})
\end{equation*}
with $f(p)= \sum_\xi F(\xi)H_{\xi} (p)$. Let us decompose $F$ in the orthogonel sum composed of subspaces ${\mathcal H}_{n,\varepsilon}$ ($\varepsilon$ are elements of $\{0,1\}^{\Z^d}$):
\begin{equation*}
F=\sum_{n,\varepsilon} F_{n,\varepsilon}, \quad F_{n,\varepsilon} \in {\mathcal H}_{n,\varepsilon}
\end{equation*} 
 Since ${\mathcal S}$ send an element of ${\mathcal H}_{n,\varepsilon}$ on an element of ${\mathcal H}_{n,\varepsilon}$, we have
 \begin{equation*}
 \| F \|_{1,\lambda}^2 = \sum_{n,\varepsilon} \| F_{n,\varepsilon}\|_{1,\lambda}^2
 \end{equation*}
 It follows that in the third line of (\ref{eq:102}), one can restrict the supremum over functions $F$ belonging to ${\mathcal H}_{2,0}$. If $F$ belongs to ${\mathcal H}_{2,0}$ then $F$ is identified with a symmetric function on $(\Z^d)^2$ and we have
 \begin{equation*}
 \|F\|_{1,\lambda}^2 \ge C \| F \|_{1,\lambda,0}^2
 \end{equation*}
 for a positive constant $C$. The supremum is then easily computed using Fourier transform. In the supremum appears the Green function of the discrete Laplacian which is finite only for dimension $d \geq 3$. We have
 \begin{equation*}
 \ll w_1, (\lambda -L)^{-1} w_1 \gg \le C, \quad d \ge 3
 \end{equation*}
 
 \section{Conclusions and final remarks}
 \label{sec:8}
 
For $d = 2$, the diffusion coefficient is expected to be of order $(\log t)^{2/3}$. It
has been proved by Yau in \cite{19} in a very technical paper for ASEP. It is not
clear that the method of \cite{19} can be applied for AEM. Indeed, constants $C(n)$ appearing in lemma \ref{lem:1} are exponential in $n$ and the method of \cite{19} seems to be restricted to polynomial dependence in $n$. Remark that ASEP is the only model belonging to KPZ class for which such behavior is proved for $d=2$. For the case $d \ge 3$, it would be interesting using generalized duality techniques to establish a fluctuation-dissipation equation (\cite{13}) for AEM, meaning a decomposition of the current in the form
\begin{equation}
w_i = \nabla \varphi_i + Lh_i
\end{equation}
for functions $\varphi_i$ and $h_i$ in a suitable Hilbert space.

In this paper we obtained lower (and upper for $d \ge 3$) bounds for the
diffusion coefficient for AEM. The strategy was based on "generalized duality
techniques" similar to \cite{3} and \cite{12}. A recent paper of Bal\'azs and Sepp\"al\"ainen
(cf. \cite{1}) improves considerably the lower bounds obtained in \cite{12} in the one dimensional case for the diffusion coefficient. For nearest neighbors ASEP (but not for general ASEP), the authors of \cite{1} are able to prove upper and lower bounds with the right order. Nevertheless, their method has restrictions : a key ingredient is attractivity of the process (AEM is not) and the method
has only been developed for the one dimensional case. Of course, generalized duality techniques have also restrictions. The key ingredients are:
\begin{itemize}
\item
If $(\mu_\rho)_{\rho}$, is a shift invariant€ family of stationary measures indexed by the
conserved quantity $\rho$ (e.g. the density), then $\LL^{2} (\mu_\rho)$ can be decomposed
in an orthogonal sum $\oplus_{n \ge 0} {\mathcal H}_n$.
\item
The symmetric part of the generator $S$ maps ${\mathcal H}_n$ into ${\mathcal H}_n$ ($S$ conserves
the degree). It means that the symmetric part of the generator restricted to ${\mathcal H}_n$ is the generator of a reversible Markov process with a finite number of particles.
\item
In general, the asymmetric part $A$ of the generator does not conserve the degree but the action of $A$ on each subspace ${\mathcal H}_n$ is a bounded operator from $({\mathcal H}_{n}, \|\cdot\|_1)$ into $\left (\oplus_{j=n-n_0}^{n+n_0} {\mathcal H}_{j}, \| \cdot \|_{-1}\right)$ where $n_0$ is a fixed positive integer.
\end{itemize} 
 
A model introduced by Sepp\"al\"ainen in \cite{17} belongs to the KPZ universality
class and should have anomalous behavior in low dimension as ASEP and AEM. Even if one consider this attractive 
process, methods of \cite{1} are difficult to apply (there is not concept of second
class particle for this process). This process is also difficult to study
with duality techniques because the symmetric part of the generator does
not conserve the degree. But a suitable modification of the process can
be studied with duality techniques. The slight modification is the discrete
counterpart of the asymmetric energy model. The symmetric part is given
by the KMP process (\cite{9}) and the asymmetric part by the asymmetric part
of the process defined in \cite{17}. The state space is $\R^{\Z^d}$ . For a real valued
local function $f(\xi)$ defined on the state space of the process, the action of the generator $L$ on $f$ is given by
\begin{eqnarray}
\label{eq:108}
(Lf)(\xi)&=& \sum_{i=1}^d \sum_{x \in \Z^d} \int_0^1 dp \left[ f(E_{x,x+e_i,p} \xi) -f(\xi)\right] +a_i \xi_x \left[ f(T_{x,x+e_i,p} \xi) - f(T_{x,x-e_i,p} \xi)\right] 
\end{eqnarray}
where the exchange operator E and the transfer operator $T$ are defined by
\begin{equation*}
\begin{cases}
E_{x,y,p} \xi = \xi +[p(\xi_x +\xi_y) - \xi_x] \delta_x +[(1-p)(\xi_x +\xi_y) -\xi_y]\delta_y\\
T_{x,y,p} \xi = \xi +p \xi_x (\delta_y -\delta_x)
\end{cases}
\end{equation*}
This Markov process is well defined, conserves the energy $\sum_x \xi_x$ and centered exponential product measures of the form
\begin{equation*}
\mu_\lambda (d\xi) = \prod_x \lambda e^{-\lambda \xi_x} d\xi_x
\end{equation*}
are invariant. The symmetric part of the generator is the KMP process. At the
difference of the symmetric part of the process defined in \cite{17}, it conserves
the degree and generalized duality technique presented here (the basis is composed of multivariate Laguerre polynomials) can be applied to the process defined by (\ref{eq:108}) and one obtains similar lower bounds for the corresponding diffusion coefficient.
 
\section{Appendix}
\label{sec:appendix}

We derive here the expression of the generator $L$ in the dual basis $H_{\xi}$. Recall that the Hermite polynomials $(h_n)_n$ satisfy the following equations:
\begin{equation}
\label{eq:hermite}
\begin{array}{l}
\cfrac{d^2 h_n}{du^2}- u \cfrac{d h_n}{du} +nh_n =0\\
\\
\cfrac{d h_n}{du} = n h_{n-1}\\
\\
h_{n+1} (u) = u h_n (u) -n h_{n-1} (u)\\
\\
h_0 =1 
\end{array}
\end{equation}
where we adopt the convention that $h_n =0$ if $n < 0$. We denote by ${\tilde H}_{\xi}$ the multivariate Hermite polynomial without the normalization factor $n(\xi)$:
\begin{equation*}
{\tilde H}_\xi = n(\xi) H_\xi
\end{equation*}

Fix $x \neq y$ in $\Z^d$. By (\ref{eq:hermite}), we have 
\begin{equation*}
(p_y \partial_{p_x} -p_x \partial_{p_y}) {\tilde H}_{\xi} = \xi_x {\tilde H}_{\xi^{x,-1,y,+1}} - \xi_y {\tilde H}_{\xi^{y,-1,x,+1}}
\end{equation*}
It follows that
\begin{eqnarray*}
(p_y \partial_{p_x} -p_x \partial_{p_y})^2 {\tilde H}_{\xi}&=& \xi_x (\xi_x -1) {\tilde H}_{\xi^{x,-2,y,+2}} -\xi_x (\xi_y +1) {\tilde H}_{\xi}\\
&+& \xi_y (\xi_y -1) {\tilde H}_{\xi^{x,+2,y,-2}} -\xi_y (\xi_x +1) {\tilde H}_{\xi}
\end{eqnarray*}
and we get finally
\begin{eqnarray*}
(p_y \partial_{p_x} -p_x \partial_{p_y})^2 {H}_{\xi}&=& (\xi_x -1) (\xi_y +2) { H}_{\xi^{x,-2,y,+2}} -\xi_x (\xi_y +1) {H}_{\xi}\\
&+& (\xi_x +2) (\xi_y -1) {H}_{\xi^{x,+2,y,-2}} -\xi_y (\xi_x +1) {H}_{\xi}
\end{eqnarray*}
because we have $n(k-2)/n(k)=1/k$. To obtain the expression for $({\mathcal S}_i F) (\xi)$, we write $f = \sum_{\xi} F(\xi) H_{\xi}$ and we have
\begin{eqnarray*}
S_i f &=& \sum_{\xi, x} F(\xi) \left\{(\xi_x -1) (\xi_{x+e_i} +2) {\tilde H}_{\xi^{x,-2,{x+e_i},+2}} -\xi_x (\xi_{x+e_i} +1) {\tilde H}_{\xi}\ \right\}\\
&+& \sum_{\xi, x} F(\xi) \left\{  (\xi_x +2) (\xi_{x+e_i} -1) {\tilde H}_{\xi^{x,+2,{x+e_i},-2}} -\xi_{x+e_i} (\xi_x +1) {\tilde H}_{\xi}\right\}
\end{eqnarray*}
Recall here that our conventions are such that $H_{\xi} =0$ if there exists $z \in \Z^d$ such that $\xi_z <0$. By a suitable obvious change of variables, we get
\begin{eqnarray*}
S_i f &=& \sum_\xi \left\{ \sum_x (\xi_x +1) \xi_{x-e_i} \left[ F(\xi^{x,+2,x-e_i,-2})- F(\xi)\right]\right\} H_\xi\\
&+& \sum_{\xi} \left\{ \sum_x (\xi_x +1) \xi_{x-e_i} \left[ F(\xi^{x,+2,x+e_i,-2})- F(\xi)\right] \right\} H_\xi
\end{eqnarray*}

The computations are similar for the asymmetric part. We have
\begin{equation*}
p_x p_y (p_y \partial_{p_x} -p_x \partial_{p_y}) {\tilde H}_{\xi}=p_x p_y (\xi_x {\tilde H}_{\xi^{x,-1,y,+1}}- \xi_y {\tilde H}_{\xi^{y,-1,x,+1}})
\end{equation*} 
By symmetry we can restrict ourselves to compute 
\begin{equation*}
p_x p_y (\xi_x {\tilde H}_{\xi^{x,-1,y,+1}})
\end{equation*}
which is equal to
\begin{eqnarray*}
\xi_x {\tilde H}_{\xi^{y,+2}} + \xi_x (\xi_y +1) (\xi_x -1) {\tilde H}_{\xi^{x,-2}}+ \xi_x (\xi_x -1) {\tilde H}_{\xi^{x,-2,y,+2}} +\xi_x (\xi_y +1) {\tilde H}_{\xi}
\end{eqnarray*}
and we obtain that
 \begin{equation*}
p_x p_y (p_y \partial_{p_x} -p_x \partial_{p_y}) {\tilde H}_{\xi}
\end{equation*}
is equal to
\begin{equation*}
\begin{array}{l}
\xi_x {\tilde H}_{\xi^{y,+2}} + \xi_x (\xi_y +1) (\xi_x -1) {\tilde H}_{\xi^{x,-2}}\\
-\xi_y {\tilde H}_{\xi^{x,+2}} - \xi_y (\xi_x +1) (\xi_y -1) {\tilde H}_{\xi^{y,-2}}\\
+(\xi_x -\xi_y) {\tilde H}_\xi\\
+\xi_x (\xi_x -1) {\tilde H}_{\xi^{x,-2,y,+2}} -\xi_y (\xi_y -1) {\tilde H}_{\xi^{y,-2,x,+2}}
\end{array}
\end{equation*}
Since ${\tilde H}_\xi =n(\xi) H_{\xi}$, we obtain
\begin{equation*}
\begin{array}{l}
p_x p_y (p_y \partial_{p_x} -p_x \partial_{p_y}) {\tilde H}_{\xi}\\
=\xi_x (\xi_y +2) H_{\xi^{y,+2}} -\xi_y (\xi_x +2) H_{\xi^{x,+2}}\\
+(\xi_y +1)(\xi_x -1) H_{\xi^{x,-2}} - (\xi_y -1)(\xi_x +1) H_{\xi^{y,-2}}\\
+(\xi_y -\xi_x) H_{\xi}\\
+(\xi_x -1) (\xi_y +2) H_{\xi^{x,-2,y,+2}} - (\xi_y -1) (\xi_x +2) H_{\xi^{x,+2,y,-2}}
\end{array}
\end{equation*}
If $f$ is a local smooth function such that $f(p)= \sum_{\xi} F(\xi)H_{\xi}$ then
\begin{eqnarray*}
A_i f = \sum_{\xi, x} F(\xi) \left\{ \xi_x (\xi_{x+e_i} +2) H_{\xi^{x+e_i,+2}} -\xi_{x+e_i} (\xi_x +2) H_{\xi^{x,+2}}\right\}\\
+\sum_{\xi,x} F(\xi) \left\{(\xi_{x+e_i}+1)(\xi_x -1) H_{\xi^{x,-2}} - (\xi_{x+e_i} -1)(\xi_x +1) H_{\xi^{{x+e_i},-2}} \right\}\\
+\sum_{\xi,x} F(\xi) \left\{ (\xi_x -1) (\xi_{x+e_i} +2) H_{\xi^{x,-2,{x+e_i},+2}} - (\xi_{x+e_i} -1) (\xi_x +2) H_{\xi^{x,+2,{x+e_i},-2}}\right\}
\end{eqnarray*}
By suitable changes of variables, we get the announced expression for ${\mathcal A}_i F$.

\end{document}